\documentclass[12pt]{iopart}

\usepackage{iopams}  

\usepackage{caption}
\usepackage{graphicx}
\usepackage{amsfonts}
\usepackage{amsthm}
\usepackage{color}

\begin{document}

\newtheorem{corollary}{Corollary}[section]
\newtheorem{lemma}{Lemma}[section]
\newtheorem{theorem}{Theorem}[section]
\newtheorem{definition}{Definition}[section]

\title[Edge Union of Networks on the Same Vertex Set]{Edge Union of Networks on the Same Vertex Set.}

\author{Chuan Wen, Loe$^1$ and Henrik Jeldtoft Jensen$^2$}

\address{Department of Mathematics and Complexity \& Networks Group\\ Imperial College London, London, SW7 2AZ, UK\\}
\ead{(1) c.loe11@imperial.ac.uk,  (2) h.jensen@imperial.ac.uk}
\begin{abstract}
Random networks generators like Erd\H{o}s-R\'{e}nyi, Watts-Strogatz and Barab\'{a}si-Albert models are used as models to study real-world networks. Let $G^1(V,E_1)$ and $G^2(V,E_2)$ be two such  networks on the same vertex set $V$. This paper studies the degree distribution and clustering coefficient of the resultant networks, $G(V, E_1 \cup E_2)$.
\end{abstract}

\maketitle

\section{Introduction}
Random networks are used as models to study real-world networks like the Internet and social networks. Some of the families of random networks generators are Erd\H{o}s-R\'{e}nyi, Watts-Strogatz and Barab\'{a}si-Albert models. These models generate  networks with only one type of edges and hence may be thought of as relating to  agents interacting under a homogeneous relation \cite{Newman:2010:NI:1809753, David:2010:NCM:1805895, Jackson:2008:SEN:1571644}.

However recently research has started to address the limitations of homogeneous networks. In the real world, agents may be related through more than one form of relations. E.g. think of social interactions between people, one type of relations may consist of family bonds another of shared interests, such as chess, a third could be professional relations and so forth. 

There are many names for networks with multiple types of connections, different researchers use different names:   Multigraph \cite{gjoka:multigraph}, Multi-Layered \cite{Piotr,Nefedov,Thiran,li}, Multi-Relational \cite{Rodriguez,journals/corr/abs-0806-2274,szell,Davis,Cai}, Multi-Dimensional \cite{hossmann12a,berlingerio:finding,{conf/icdm/RossettiBG11},conf/asunam/BerlingerioCGMP11,tang,tang2}, Multiplex \cite{szell,Mucha,Cardillo,Gligor,Brummitt}, Multi-Modal \cite{Heath} or graphs on the same vertex set \cite{journals/cpc/KuhnO07}. We will here use the term multi-dimensional.
We may think of these networks as compose of layers of different relations (dimensions). Each layer has the same set of vertices, and in every layer the vertices are connected through different relationships.

There are various proposals to solve classic problems like link prediction and community detection in multi-dimensional networks. However, so far the structural property of the entire network has not been investigated. Given the statistical properties of two layers, we address in this paper the statistical properties of the union of the different types of edges. Combining edges of different types is known as the \emph{edge union of networks on the same vertex set}.

We note that the famous Zachary Karate Club Network \cite{zachary1977ifm} is an example of an edge union of 8 networks on the same vertex set. Any two members (vertices) of the karate club are connected if they have consistent interactions outside classes and club meetings. ``Outside interactions" are defined in terms of  8 different types of relationships. One of such relations is the association in and between academic classes at the university. The 8 different relationships are then combine as edges (of different types) on the same vertex set. Other real-world applications of these kind of networks can be found in \cite{KimL10-3} and \cite{0295-5075-71-1-152}.

As another example of the relevance of multi-dimensional networks think of the network of email exchanges between people. Assume that the preferential attachment of the Barab\'{a}si-Albert model is true in this context, where a new member tends to email the more popular members. However being ``popular" at a personal level is different from being ``popular" in the work environment. If the database of email exchange does not distinguish the two contexts, the email exchange network is essentially the union of two Barab\'{a}si-Albert networks.

The present paper is a study of some of the statistical properties of a random network that is formed by the union of two random networks. The different families of random network are distinguished by their clustering coefficients and the degree distribution of the vertices and we want to determine the corresponding measures for the union of different types of networks. We first consider the degree distributions and then the clustering coefficients. 

\section{Network Families}
This section summarises the construction of Erd\H{o}s-R\'{e}nyi, Watts-Strogatz and Barab\'{a}si-Albert models. Recall that these random network generators were designed to construct simple graphs. We are going to combine networks with different edge statistics but with the same number of vertices, accordingly we suppress in the notation below the number of vertices in the labelling of the specific network types. 
For instance, the usual notation for Gilbert Graph is $G_{n,p}$, where $n$ is the size of the vertex set. We simplify the notation to $G_p$ as the size of the vertex set is implied to be $n$ for all the networks in the rest of the paper. Lastly, let $V$ and $E$ be vertex set and edge set respectively.

\subsection{Erd\H{o}s-R\'{e}nyi}
A realisation of  an Erd\H{o}s-R\'{e}nyi network \cite{citeulike:4012374} is selected with equal probability from the set of all possible graphs with $n=|V|$ vertices and $|E|$ edges. However to generate a huge  random Erd\H{o}s-R\'{e}nyi is difficult. To circumvent this problem one may instead let the number of edges  fluctuate slightly and consider a  Gilbert graph \cite{Gilbert1959} $G_{p}$ in which every vertex pair is connected with probability $p$ where $p \approx |E|/{ |V|  \choose 2}$. The slight difference is that Erd\H{o}s-R\'{e}nyi has precisely $|E|$ edges while Gilbert graph has approximately $|E|$ edges with high probability.

\subsection{Watts-Strogatz}
A Watts-Strogatz network \cite{1998Natur.393..440W}, $W_{w,q}$ is parameterised by $w$ and $q$ for mean degree\footnote{We replace the commonly used variable $k$ with $w$ to avoid confusion with the variable $k$ commonly used in degree distribution function} and probability of rewiring respectively. The construction begins with a regular ring lattice where each vertex connects to $w/2$ neighbours on each side.

For each vertex $v_i \in V$ and $i<a$, each edge leaving $v_i$ is rewired with probability $q$. The rewiring replaces $\{v_i, v_a\}$ with $\{v_i, v_b\}$ where $v_b$ is chosen uniformly in $V$ and the resultant network remains a simple network.

The key property of the Watts-Strogatz network is that it varies between a regular ring lattice ($q=0$) and an Erd\H{o}s-R\'{e}nyi network ($q=1$). As $q \rightarrow 1$, Watts-Strogatz can be expressed as a Erd\H{o}s-R\'{e}nyi graph, i.e. $W_{w,q} \rightarrow G_{w/(n-1)}$.

\subsection{Barab\'{a}si-Albert}
 Barab\'{a}si-Albert network \cite{barabasia99} is denoted by $B_m$ where $m$ is the number of new edges at each iteration. The network construction begins with some arbitrary small number of vertices connected randomly. 

At each iteration, one new vertex of degree $m$ is added. The edges of the new vertex are connected probabilistically with a probability proportional to the degree of the existing vertices. Define $deg(v_i)$ as the degree of vertex $v_i$. The probability that the new vertex is connected to vertex $v_i$ is given by:
\begin{equation*}
p_i = \frac{deg(v_i)}{\sum_j deg(v_j)}.
\end{equation*}
This is referred to as preferential attachment.

\section{Definitions and Preliminaries}
\subsection{Definitions}
For brevity, the resultant network from the edge union of networks is defined as the \emph{composite network} (abbreviation: $\mathcal{C}$). A \emph{network} refers to either an Erd\H{o}s-R\'{e}nyi, Watts-Strogatz or Barab\'{a}si-Albert network (abbreviation: ER, WS, BA respectively). The composition of all 6 pairwise edge unions of these networks are $\{ER, ER\}$, $\{ER, WS\}$, $\{ER,BA\}$, $\{WS, WS\}$, $\{WS, BA\}$ and $\{BA, BA\}$.

\begin{definition}
Let $G^1(V,E_1)$ and $G^2(V,E_2)$ be two networks on the same vertex set $V$. The \textbf{edge union of the networks} gives the composite network, $\mathcal{C}(V, E_1 \cup E2)$.
\end{definition}

\begin{definition}
Let $G^1(V,E_1)$ and $G^2(V,E_2)$ be two  networks on the same vertex set $V$. An edge $e \in E_1 \cup E_2$ in the composite network is called a \textbf{common edge} if $e \in E_1 \cap E_2$. Hence the set of common edges is $E_1 \cap E_2$.
\end{definition}

\begin{center}
  \includegraphics[width=130mm]{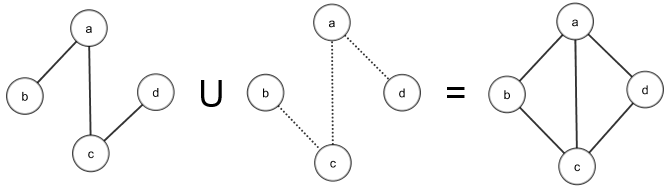}
    \captionof{figure}{Edge union of two networks on the same vertex set. The solid edges and the dotted edges are from the edge set of different networks. The composite network is the rightmost network. Edge $\{a,c\}$ is a common edge between the two networks.}\label{graph_example}
\end{center}

\begin{definition}
Consider a population of $n$  urns. Choose $a<n$ urns and place one red ball in each. Next choose $b<n$ urns and place one blue ball in each. The probability that exactly $i$ urns  contain one red and one  blue ball is given the hypergeometric function
\begin{equation*}
\mathcal{H}(i;n,a,b) = \frac{{a \choose i} {n-a \choose b-i}}{{n \choose b}}
\end{equation*}
\end{definition}

\subsection{Preliminaries}
The following lemmas and concepts are useful in the later analysis.

\begin{lemma}\label{thm:ProbOverlap}
Let $E_1$ and $E_2$ be the sets of edges of two networks with $n$ vertices each. The probability that there are $\epsilon$ common edges is:
\begin{equation}\label{ProbOverlap}
P(|E_1 \cap E_2|=\epsilon) = \mathcal{H}(\epsilon; {n \choose 2},|E_1|,|E_2|)
\end{equation}
\end{lemma}
\begin{proof}
Choose $|E_2|$ edges out of a total possible $n \choose 2$ edges and colour them blue. Next choose $|E_1|$ edges from the set $E_1$, the probability of getting $\epsilon$ blue edges is defined by the hypergeometric function.
\end{proof}

We note that equation \ref{ProbOverlap} is the probability that there are $\epsilon$ common edges between the two networks. In addition, the expected number of common edges is given by:

\begin{corollary}\label{lemma:ExpectedOverlap}
Let $E_1$ and $E_2$ be the set of edges of two networks with $n$ vertices each. The expected number of common edges is the expectancy of the hypergeometric function in equation \ref{ProbOverlap}:

\begin{equation}
\mathbb{E}[|E_1 \cap E_2|] = \mathbb{E}[\mathcal{H}] = \frac{|E_1| \cdot |E_2|}{{n \choose 2}}
\end{equation}
\end{corollary}

In fact we can generalise the above to compute the number of common cliques in the union. In particular it is useful to count the number of common triangles (3-clique) as it affects the accuracy of the clustering coefficient in the union to two networks.

\begin{lemma}\label{thm:ProbOverlapClique}
Let $K_1$ and $K_2$ be the sets of c-cliques of two networks with $n$ vertices each. The probability that there are $\epsilon$ common c-cliques is:
\begin{equation}\label{ProbOverlapClique}
P(|K_1 \cap K_2|=\epsilon) = \mathcal{H}(\epsilon; {n \choose c},|K_1|,|K_2|)
\end{equation}
\end{lemma}
\begin{proof}
Choose $|K_2|$ c-cliques out of a total possible $n \choose c$ c-cliques and colour them blue. Next choose $|K_1|$ c-cliques from the set $K_1$, the probability of getting $\epsilon$ blue c-cliques is defined by the hypergeometric function.
\end{proof}

\begin{corollary}\label{lemma:ExpectedOverlapClique}
Let $K_1$ and $K_2$ be the sets of c-cliques of two networks with $n$ vertices each. The expected number of common c-cliques is the expectancy of the hypergeometric function in equation \ref{ProbOverlapClique}:

\begin{equation}
\mathbb{E}[|K_1 \cap K_2|] = \mathbb{E}[\mathcal{H}] = \frac{|K_1| \cdot |K_2|}{{n \choose c}}
\end{equation}
\end{corollary}

To compute the degree distribution of the composite network accurately, we have to minimise any double counting of the common edges. For example in Fig. \ref{graph_example}, vertex $a$ of the composite network has degree 3. It is the sum of the degree of the networks (degree 2 from each network) \textbf{minus} the number of common edges ($=1$) at vertex $a$.

\begin{lemma}\label{lemma:VertexProbOverlap}
Let $v_{i,G^1}$ and $v_{i,G^2}$ be the vertices of networks $G^1$ and $G^2$ each with $n$ vertices respectively. Let $d_1=deg(v_{i,G^1})$ and $d_2=deg(v_{i,G^2})$, the probability that there are $\epsilon$ common edges between $v_{i,G^1}$ and $v_{i,G^2}$ is:

\begin{equation}\label{eq:VertexProbOverlap}
P_{ce}(\epsilon | d_1, d_2) = \mathcal{H}(\epsilon;n-1,d_1,d_2)
\end{equation}
\end{lemma}

\begin{proof}
There are only $n-1$ vertices left for $v_{i,G^1}$ and $v_{i,G^2}$ to connect to. Using the same argument as in Lemma \ref{thm:ProbOverlap}, there are $d_1$ blue edges and $d_2$ red edges from $v_{i,G^1}$ and $v_{i,G^2}$ respectively.
\end{proof}

\begin{corollary}\label{lemma:ExpectedOverlapAtVertex}
Let $v_{i,G^1}$ and $v_{i,G^2}$ be the vertices of networks $G^1$ and $G^2$ with $n$ vertices respectively. Let $d_1=deg(v_{i,G^1})$ and $d_2=deg(v_{i,G^2})$, the expected number of common edges between $v_{i,G^1}$ and $v_{i,G^2}$ is:

\begin{equation}
\mathbb{E}[\mathcal{H}] = \frac{|d_1| \cdot |d_2|}{n-1}
\end{equation}
\end{corollary}

In this paper we only consider the case where both networks have approximately the same edge set size. Unequal edge sets size introduces additional complexity to the study, since the larger edge set dominates the resultant network.

\section{Degree Distribution of Edge Union of Networks}\label{sec:DD}

\subsection{Erd\H{o}s-R\'{e}nyi with Erd\H{o}s-R\'{e}nyi}\label{ER2ER}
Let $G^1_{p}$ and  $G^2_{p'}$ be two independent Gilbert Graphs. The probability that an edge does not exist in both $G^1_{p}$ and $G^2_{p'}$ is $(1-p)(1-p')$, accordingly the probability that the composite network will contain an edge is given by $1-(1-p)(1-p')$ and we have
\begin{equation}\label{eq:ER2ER}
G^1_{p} \cup G^2_{p'} = G_{1-(1-p)(1-p')}
\end{equation}

\subsection{Erd\H{o}s-R\'{e}nyi with Watts-Strogatz}\label{ER2WS}
Let $G_{p}$ and $W_{w,q}$ be an ER network and WS network respectively. As $q \rightarrow 1$, $W_{w,q}$ can be approximated as an ER so in this limit the union will be scribed by Eq. (\ref{eq:ER2ER}).

For $q \rightarrow 0$, the degree distribution of the WS network approaches $P^{W}(k)=\delta(k-w)$, the number of vertices with degree $k$. Neglecting common edges, a vertex in the composite network will have degree $k$ when the corresponding vertex in ER has degree $k-w$. Hence an approximation to the degree distribution of the composite network is:

\begin{equation}\label{ER2WS_dd}
P^{\mathcal{C}}(k) \sim \frac{(np)^{(k-w)} e^{-np}}{(k-w)!}
\end{equation}

We can improve the estimate of the degree distribution by explicitly discounting the common edges in the union. This will become of increasing importance as $q$ is increased from 0. Let the degree distribution of ER and WS be $P^{G}(k)$ and $P^{W}(k)$ respectively. To account for the common edges for $P^\mathcal{C}(k)$, we have to consider:

\begin{enumerate}
\item Probability that a vertex in WS, $v_{ws}$ has degree $j$;
\item Probability that a vertex in ER, $v_{er}$ has degree $k+ \epsilon -j$;
\item Probability that there are $\epsilon$ common edges between $v_{er}$ and $v_{ws}$.
\end{enumerate}

From lemma \ref{lemma:VertexProbOverlap}, the probability that there are $\epsilon$ common edges between $v_{er}$ and $v_{ws}$ is $P_{ce}(\epsilon | j,k+ \epsilon -j)$. Thus the improvement on equation \ref{ER2WS_dd} is:

\begin{equation}\label{ER2WS_P}
P^{\mathcal{C}}(k) \sim \sum_\epsilon^k \sum_j^k P_{ce}(\epsilon | j,k+ \epsilon -j) P^{W}(j) P^{G}(k+ \epsilon -j)
\end{equation}

Since WS is lattice-like, most of the WS vertices have  a degree close to $w$. By assuming $P^W(k) = \delta(k-w)$  we can simplify equation \ref{ER2WS_P} to:

\begin{equation}
P^{\mathcal{C}}(k) \sim \sum_\epsilon^w P_{ce}(\epsilon | w,k+ \epsilon -w) P^{G}(k+ \epsilon -w)
\end{equation}

\subsection{Erd\H{o}s-R\'{e}nyi with Barab\'{a}si-Albert}\label{ER2BA}
It has been noticed that the degree distribution of the BA model often does not represent natural occurring networks. This has lead to various generalised BA models \cite{Pennock02winnersdont, dorogovtsev2000structure, dorogovtsev2001anomalous} in which preferential and random uniform attachment are combined. This construction is similar to the union of ER with BA, where ER adds some uniform attachment to BA.

To compute the degree distribution, we can use the method from section \ref{ER2WS}:

\begin{equation}\label{ER2BA_P}
P^{\mathcal{C}}(k) \sim \sum_\epsilon^k \sum_j^k P_{ce}(\epsilon | j,k+ \epsilon -j) P^{G}(j) P^{B}(k+ \epsilon -j)
\end{equation}
where the approximate degree distribution of BA is $P^B(k) \sim 2m^2 k^{-3}$ \cite{PhysRevE.65.057102}.

However the hypergeometric function, $P_{ce}(\epsilon | j,k+ \epsilon -j)$ is computationally expensive. It is difficult to calculate the sums in Eq.  (\ref{ER2BA_P})  for large values of 
  $k$ and thereby assess how the additional uniform attachment influence the power law form of the BA network. To extract the asymptotic behaviour, we now make use of a Fokker-Planck approach. 
The idea is to begin with the BA network and then in each time step add a new uniformly drawn random edge to the BA network. The new edges are taken from the set of ER edges not common with the BA edge set. Since there are $nm$ edges in BA,  from corollary \ref{lemma:ExpectedOverlap} it follows that there are $nmp$ common edges or ${n \choose 2}p - nmp$ non-common edges.

Let $u(k, t)$ be the number of vertices of degree $k$ at time step $t$. At each time step, a new edge will change the degree of 2 vertices. Specifically the number of degree $k$ vertices increases by one if the new edge attaches to a degree $k-1$ vertex. This has probability $u(k-1, t)/n$. Similarly the number of $k$ vertices will decrease by one if the new edge attaches to a degree $k$ vertex. This leads to the following equation,

\begin{equation}\label{iterative_ER2BA_dd}
u(k, t+1) = u(k,t) + u(k-1, t)/n - u(k, t)/n
\end{equation}

By replacing $t$ and $k$ by continuous variables, we obtain a partial differential equation which will be a good approximation for large values of $k$.

\begin{equation}
\frac{\partial u}{\partial t} + \frac{1}{n} \frac{\partial u}{\partial k} =  0. 
\label{continuum} 
\end{equation}

Using the initial condition $u(k,0) = P^B(k) \cdot n \sim 2nm^2 k^{-3}$ ($P^B(k) = 0$ for $k < m$), we can solve for $u(k,t)$ at $t=2 \cdot {n \choose 2}p - nmp \approx np(n-m)$ (twice the number of non-common edges because there are two vertices that change at each time step) and find,
\begin{equation}
P^{\mathcal{C}}(k) = u(k, t) = \frac{2nm^2}{(k + 2n(p-p^2))^3}
\label{asymptotic}
\end{equation}
In Fig. \ref{ER2BA_degree} we compare this asymptotic expression to the analytic expression derived by iterating Eq. ({\ref{iterative_ER2BA_dd}) and to simulation results. The continuum approximation of Eq. (\ref{continuum}) does not really have a region of validity since the finite number of vertices limits the size of the degree and the asymptotic limit cannot be reach. However, the iterated solution of the Fokker-Planck equation (\ref{iterative_ER2BA_dd}) matches the simulations well.

\begin{center}
  \includegraphics[width=130mm]{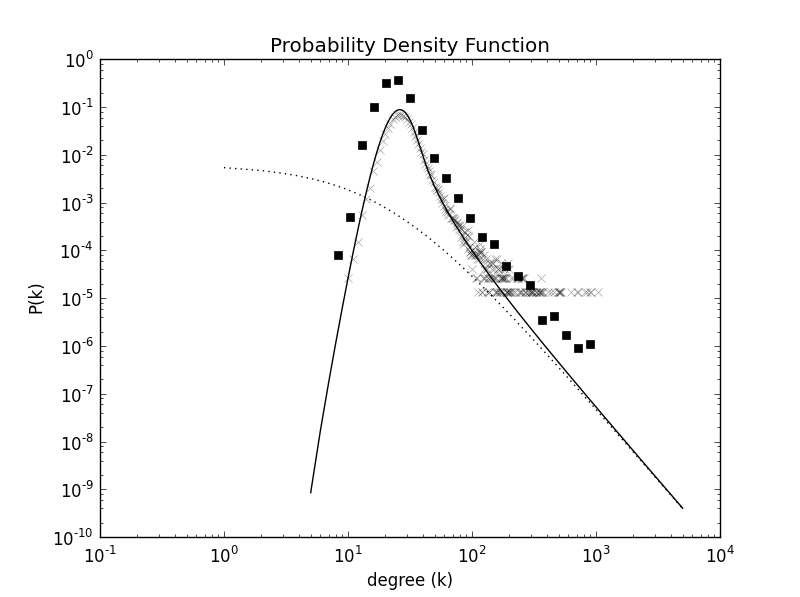}
    \captionof{figure}{Parameters: $n=75000$, $m=10$ and $p=2m/(n-1)$ (chosen such that both networks have equal number of edges) .The solid line plots the iterative method, Eq. (\ref{iterative_ER2BA_dd}). The dotted line plots the closed form Eq. (\ref{asymptotic}). The crosses represent the distribution obtained from simulations. The degree-distribution log-binning of the crosses is plotted with squares.}
\label{ER2BA_degree}   
\end{center}
 
\subsection{Watts-Strogatz with Watts-Strogatz}\label{WS2WS}
Let $W^1_{w,q}$ and $W^2_{w',q'}$ be two WS networks. If $q \approx q' \rightarrow 1$, the two WS can be expressed as two ER (section \ref{ER2ER}). In contrast the limit $q \rightarrow 1$ and $q' \rightarrow 0$  is identical to the union of ER and WS considered in section \ref{ER2WS}. We do not discuss the case for general values of $q$ and $q'$.

\subsection{Watts-Strogatz with Barab\'{a}si-Albert}\label{WS2BA}
Let $W_{w,q}$ and $B_{m}$ be WS and BA graphs respectively. There are $wn/2$ edges in WS and $nm$ edges in BA. It is most interesting to consider the case when the graphs have approximately equal number of edges, i.e. $w \approx 2m$. We ignore the case $q \rightarrow 1$, as it is similar to Sec. \ref{ER2BA}.

For $q \rightarrow 0$, the degree distribution of $\mathcal{C}$ is straight forward. Since the probability of rewiring is low, most of the lattice edges in WS remain unchanged. Hence most of the vertices in WS have degree $w$, and in turn contributes to an increase in the degree of the vertices of BA by $w$. This gives the approximation for:

\begin{equation}\label{eq:WS2BA_dd}
P^{\mathcal{C}}(k) \sim 2m^2 (k-w)^{-3}
\end{equation}

The degree distribution can be refined by considering the common edges between WS and BA. However the expression will be in open form and provides little insights. From simulations ($q \rightarrow 0$), Eq. \ref{eq:WS2BA_dd} is sufficient to approximate the distribution (Fig. \ref{WS2BA_dd}).  
\begin{center}
  \includegraphics[width=130mm]{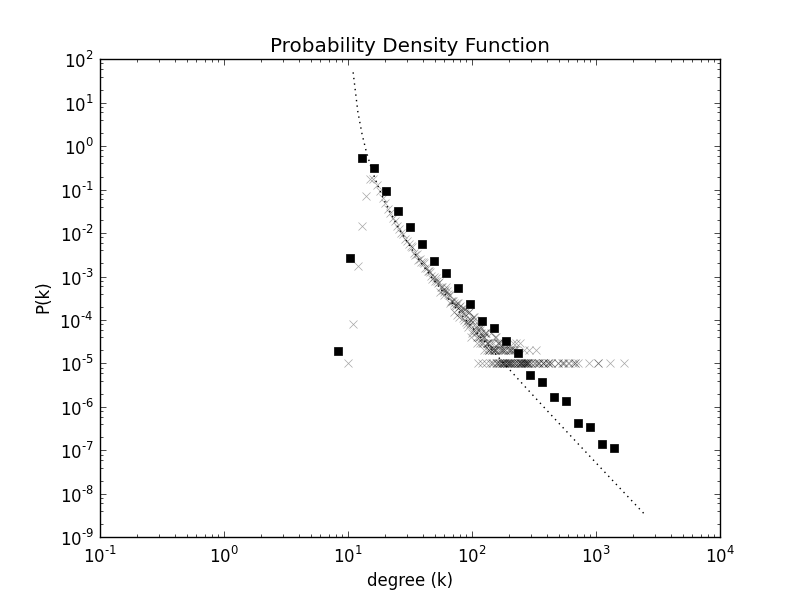}
    \captionof{figure}{Parameters: $n=100000$, $w=10$, $m=5$ and $q = 0.1$. The crosses represents the degree distribution obtained from simulations. The dotted line plots Eq. \ref{eq:WS2BA_dd}. The degree-distribution log-binning of the crosses is plotted with squares.} \label{WS2BA_dd}
\end{center}
 
\subsection{Barab\'{a}si-Albert with Barab\'{a}si-Albert}\label{BA2BA}
The key feature of BA is that the degree distribution is scale-free, i.e. follows a power-law in the degree. Hence we want to know if the union of two BA retain the features. Similar to equation  \ref{ER2WS_P}, the probability density function is the combined probability of:

\begin{equation}\label{eq:BA2BA_P}
P^{\mathcal{C}}(k) \sim \sum_\epsilon^k \sum_j^k P_{ce}(\epsilon | j,k+ \epsilon -j)P^{B}(j) P^{B}(k+ \epsilon -j)
\end{equation}

Unfortunately, we are not able to simplify equation \ref{eq:BA2BA_P} to derive the asymptotic behaviour. However simulations (figure \ref{BA2BA_dd}), indicate a power-law distribution for large $k$.

\begin{center}
  \includegraphics[width=130mm]{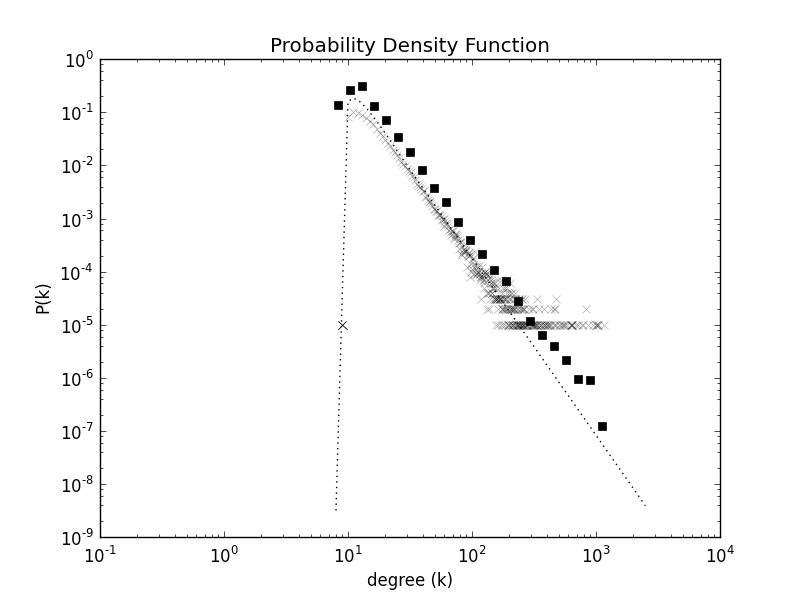}
    \captionof{figure}{The dashed line represents Eq. \ref{eq:BA2BA_P} for $n=100000$ and $m=5$. The crosses represent the results of simultions. The degree-distribution log-binning of the crosses is plotted with squares.}\label{BA2BA_dd}
\end{center}

 Unlike in Sec. \ref{ER2BA}, we are not able to apply the same method to the union of two BA networks. In Sec. \ref{ER2BA} we began with a BA and then we added uncorrelated random edges onto it. In the present case the preferential attachment induces correlations which cannot be neglected. To illustrate this better, let the first BA have red edges, and the second BA to have blue edges. Begin with the BA with red edges, and assume we can put the first few blue edges randomly like ER. However the subsequent blue edges have a preferential bias for {\bf only the blue} edges. Since u(k,t) in section \ref{ER2BA} does not differentiate the colours on the edges, the method does not apply.

Figure \ref{DD_all} shows the degree rank plot of different combinations of the network union. It is clear that the distribution of ER/BA is slightly different from BA/BA. The union of BA/BA has more high degree vertices than those of BA/ER.

\section{Clustering Coefficients of Edge Union of Networks}
Let us now consider the clustering coefficient of the composite network. While we are not always able to analytically determine the clustering coefficient, we are able to derive lower and upper bounds.  Again we assume that there are few common edges/triangles in the union, which can be verified with corollary \ref{lemma:ExpectedOverlap} and corollary \ref{lemma:ExpectedOverlapClique}. This simplification is reasonable for sparse networks and it gives us simple relationships similar to Eq. (\ref{eq:WS2WS}).

\subsection{Watts-Strogatz with Watts-Strogatz}\label{WS2WS_Cluster}
The clustering coefficient of such a union can be computed analytically in the limit $p \approx q$ and $p, q \rightarrow 0$. Define $T(v_i)$ as the number of subgraphs with 3 edges and 3 vertices (triangles) which all share the vertex $v_i$.   Recall the clustering coefficient equation for a network:

\begin{eqnarray}\label{eq:cc}
\mbox{Clustering Coefficient} &=& \frac{1}{n} \sum_{i=0}^n \mbox{Clustering Coefficient of } v_i \\
 &=& \frac{1}{n} \sum_{i=0}^{n} \frac{T(v_i)}{{deg(v_i) \choose 2}}\\
\end{eqnarray}

Let $v_{i,w^1}$ and $v_{i,w^2}$ be the vertices of $W^1$ and $W^2$ respectively, which are mapped to the $i^{th}$ vertex ($v_{i,c}$) of $\mathcal{C}$. Since $p \approx q$ and $p, q \rightarrow 0$, almost every vertex of $W^1$ and $W^2$ are similar in structure. i.e. $T(v_{i,w^1}) \approx T(v_{i,w^2})$ and $deg(v_{i,w^1}) \approx deg(v_{i,w^2})$ Hence even with random pairing, every vertex of the composite network $\mathcal{C}$ will be similar. We therefore have for small $w$:

\begin{eqnarray*}
\mbox{Clustering Coefficient of } \mathcal{C} &=& \frac{1}{n} \sum_{i=0}^n \mbox{Clustering Coefficient of vertex } v_{i,c}  \\
  &=& \frac{1}{n} \sum_{i=0}^{n} \frac{T(v_{i,c})}{{deg(v_{i,c}) \choose 2}}
\end{eqnarray*}
 
Because $w$ is small ($w << n$), we assume there are few common edges. Hence $T(v_{i,c}) = T(v_{i,w^1}) + T(v_{i,w^2}) \approx 2T(v_{i,w^1})$ and $deg(v_{i,c}) \approx 2 deg(v_{i,w^1})$. Then:

\begin{eqnarray}
\mbox{Clustering Coefficient of } \mathcal{C} &\approx& \frac{1}{n} \sum_{i=0}^{n} \Bigg( \frac{2 T(v_{i,w^1})}{{2deg(v_{i,w^1}) \choose 2}}\Bigg) \label{cc_ws2ws_approx} \nonumber \\
  &\approx&  \frac{1}{n} \sum_{i=0}^{n} \frac{1}{2} \Bigg( \frac{T(v_{i,w^1})}{{deg(v_{i,w^1}) \choose 2}}\Bigg) \nonumber \\
&\approx& \frac{1}{2} \mbox{ Clustering Coefficient of } W^1 \label{eq:WS2WS}
\end{eqnarray}

Fig. \ref{WS2WS_cc} compares the results of Eq. (\ref{eq:WS2WS}) with actual simulation of the edge union of two WS. The dots plot the expected clustering coefficient from Eq. (\ref{eq:WS2WS}) by halving the clustering coefficient of one of the WS. The crosses represent the actual clustering coefficient obtained from simulations. The details of the simulations are explained in the appendix.

\begin{center}
  \includegraphics[width=130mm]{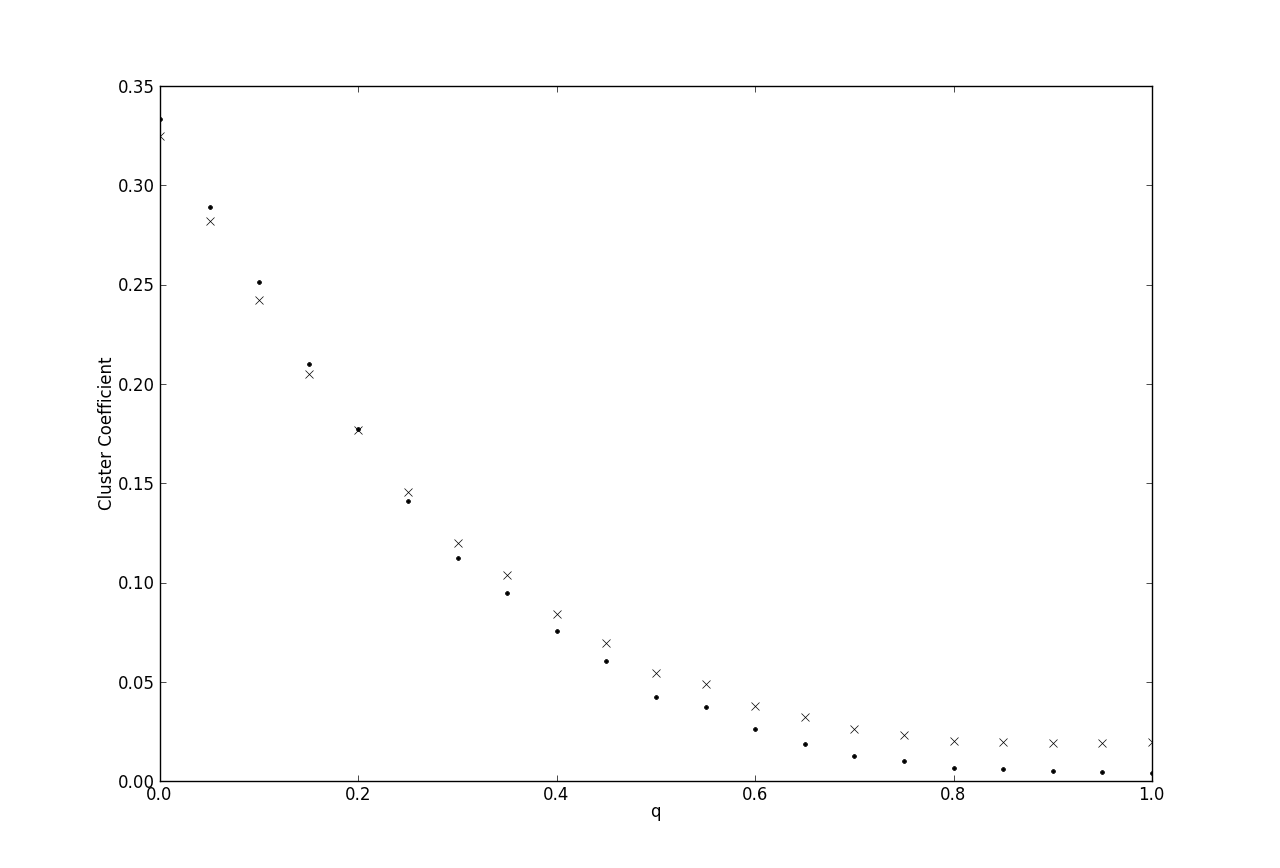}
    \captionof{figure}{Edge Union of two WS networks for $n=1000$, $w'=w=20$ and $q' = q$. The dots are from Eq. (\ref{eq:WS2WS}) and the crosses are from the simulations.}\label{WS2WS_cc}
\end{center}
 
\subsection{Watts-Strogatz with Barab\'{a}si-Albert}\label{WS2BA_Cluster}
The clustering coefficient of the composite is difficult to approximate analytically as there is no simple way to approximate the number of triangles for BA. Instead, we derive a lower bound and an upper bound for the composite:

\begin{eqnarray*}
\mbox{Clustering Coefficient of } \mathcal{C} &=& \frac{1}{n} \sum_{i=0}^n \mbox{Clustering Coefficient of vertex } v_{i,c}  \\
  &=& \frac{1}{n} \sum_{i=0}^{n} \frac{T(v_{i,c})}{{deg(v_{i,c}) \choose 2}}
\end{eqnarray*}

Assume $T(v_{i,c}) \approx T(v_{i,ws}) + T({v_{i,ba}})$.

\begin{eqnarray*}
\mbox{Clustering Coefficient of } \mathcal{C} &\approx& \frac{1}{n} \sum_{i=0}^{n} \frac{T(v_{i,ws}) + T(v_{i,ba})}{{deg(v_{i,c}) \choose 2}} \\
&\geq& \frac{1}{n} \sum_{i=0}^{n} \frac{T(v_{i,ws})}{{deg(v_{i,c}) \choose 2}}
\end{eqnarray*}

Since we are considering $q \rightarrow 0$, we can make the following two observations: 1) The number of triangles generated by WS is generally larger than BA, hence choosing $T(v_{i,ws})$ will get a tighter bound. 2) Most of the vertices in WS have the same number of triangles attached due to the assumed low rewiring probability. Let $\alpha$ be the average number of triangles for any vertex in WS given $q$. Thus:

\begin{eqnarray}\label{approx_lower}
\mbox{Clustering Coefficient of } \mathcal{C} &\geq& \frac{1}{n} \sum_{i=0}^{n} \frac{T(v_{i,ws})}{{deg(v_{i,c}) \choose 2}} \\
  &\approx& \frac{\alpha}{n} \sum_{i=0}^{n} \frac{1}{{deg(v_{i,c}) \choose 2}} \\
  &=& \alpha \sum_{k=0}^{n} P^{\mathcal{C}}(k) \frac{1}{{k \choose 2}} 
\end{eqnarray}

From \cite{barrat00}, the clustering coefficient of WS decreases at the rate $(1-q)^3$ and  $\alpha \approx (w/2)^2$ for $q=0$. This is also equal to the rate of decrease to the number of triangles since $(1-q)^3$ is the probability that none of the edges of a triangle is rewired. Thus we have the following lower bound on the clustering coefficient of the composite network:

\begin{equation}\label{eq:WS2BA_cc}
\mbox{Clustering Coefficient of } \mathcal{C} \geq \frac{(1-q)^3 w^2}{4} \sum_{k=0}^{n} P^{\mathcal{C}}(k) \frac{1}{{k \choose 2}}
\end{equation}

\begin{center}
  \includegraphics[width=130mm]{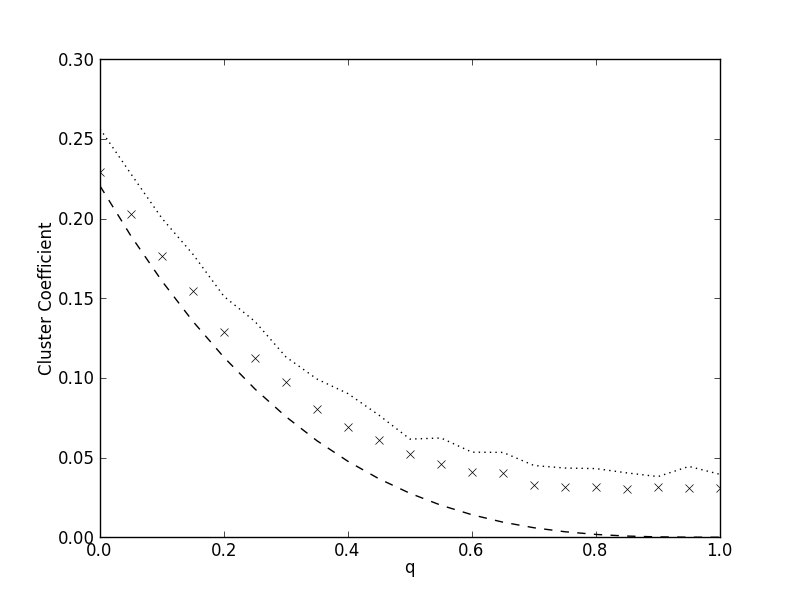}
    \captionof{figure}{Parameters: $n=1000$, $w=10$ and $m=5$. The crosses plot the  clustering coefficient from the simulation of the composite network. The dashed line is the lower bound from Eq. \ref{eq:WS2BA_cc}. The dotted line is the upper bound from Eq. \ref{eq:WS2BA_cc_U}}\label{WS2BA_cc}
\end{center}

To obtain an upper bound on the clustering coefficient, we can round up the clustering coefficient contribution originating from BA. I.e.

\begin{eqnarray}
\mbox{Clustering Coefficient of } \mathcal{C} &\approx& \frac{1}{n} \sum_{i=0}^{n} \frac{T(v_{i,ws}) + T(v_{i,ba})}{{deg(v_{i,c}) \choose 2}} \nonumber \\
&=& \frac{1}{n} \sum_{i=0}^{n} \frac{T(v_{i,ws})}{{deg(v_{i,c}) \choose 2}} + \frac{1}{n} \sum_{i=0}^{n} \frac{T(v_{i,ba})}{{deg(v_{i,c}) \choose 2}} \nonumber \\
&\leq& \frac{1}{n} \sum_{i=0}^{n} \frac{T(v_{i,ws})}{{deg(v_{i,c}) \choose 2}} + \frac{1}{n} \sum_{i=0}^{n} \frac{T(v_{i,ba})}{{deg(v_{i,ba}) \choose 2}} \nonumber \\
&\approx& \frac{(1-q)^3 w^2}{4} \sum_{k=0}^{n} P^{\mathcal{C}}(k) \frac{1}{{k \choose 2}} \nonumber \\
&& + \mbox{Clustering Coefficient of } B_m. \label{eq:WS2BA_cc_U}
\end{eqnarray}

\subsection{Erd\H{o}s-R\'{e}nyi with Erd\H{o}s-R\'{e}nyi}\label{ER2ER_Cluster}
From section \ref{ER2BA}, the union of two ER is an ER network. For simplification, let the composite network be $G^c_p$. Hence the clustering coefficient of $G^c_p$ can be estimated like any ER network.

However the ``clustering coefficient" of ER is often measured by an alternative similar metric, Transitivity Ratio \cite{newman_transitivity}. It is the ratio of the number of triangles to the number of connected triples:

\begin{equation}\label{eq:transitivity}
\mbox{Transitivity Ratio of }G^c_p = \frac{{n \choose 3}p^3}{{n \choose 3}p^2} = p.
\end{equation}

\subsection{Clustering coefficient of other combinations}\label{Others_Cluster}
In the next section, we show that the clustering coefficient of ER$\cup$BA and BA$\cup$BA are small in general (Fig. \ref{CC_all}). For this reason we don't derive any analytical results in this paper.

WS$\cup$ER is a more interesting combination, but is difficult to deduce the clustering coefficient of the union from the know results for the individual families of networks. The clustering coefficient of WS is measured by Eq. \ref{eq:cc} and the ``clustering coefficient" (transitivity ratio) of ER is measured by Eq. \ref{eq:transitivity}. Since these metrics are different in general \cite{Schank}, we are only able to deduce a weak estimate of the union (Fig. \ref{WS2ER_cc}).

Firstly in the composite network, the maximum number of triangles at vertex $v_{i,c}$ is ${deg(v_{i,c})\choose2}$. Among them, $T(v_{i,ws})$ triangles are from WS. The rest of the triangles exist if the edges of ER connect the neighbors of $v_{i,c}$ that are not connected from the edges of WS.

With probability $p$ from ER, the number of neighbors pairs at $v_{i,c}$ that are connected by the edges of ER (and not from WS) is:

\begin{equation*}
\Big({deg(v_{i,c})\choose2} - T(v_{i,ws}) \Big)p
\end{equation*}

Hence, 
\begin{eqnarray}
\mbox{Clustering Coefficient of } \mathcal{C} &\approx& \frac{1}{n} \sum_{i=0}^{n} \frac{\Big({deg(v_{i,c})\choose2} - T(v_{i,ws}) \Big)p + T(v_{i,ws})}{{deg(v_{i,c}) \choose 2}} \nonumber \\
&=& \frac{1}{n} \sum_{i=0}^{n} \frac{{deg(v_{i,c})\choose2}p +T(v_{i,ws})(1-p)}{{deg(v_{i,c}) \choose 2}} \nonumber \\
&=& p + \frac{(1-p)}{n} \sum_{i=0}^{n} \frac{T(v_{i,ws})}{{deg(v_{i,c}) \choose 2}} \nonumber
\end{eqnarray}

Together with approximation from Eq. \ref{approx_lower} and Eq. \ref{eq:WS2BA_cc},
\begin{eqnarray}
\mbox{Clustering Coefficient of } \mathcal{C} &=& p + \frac{(1-p)}{n} \sum_{i=0}^{n} \frac{T(v_{i,ws})}{{deg(v_{i,c}) \choose 2}} \nonumber \\
&\approx& p + \frac{(1-p)(1-q)^3 w^2}{4} \sum_{k=0}^{n} P^{\mathcal{C}}(k) \frac{1}{{k \choose 2} \label{eq:WS2ER_cc}}
\end{eqnarray}

\begin{center}
  \includegraphics[width=120mm]{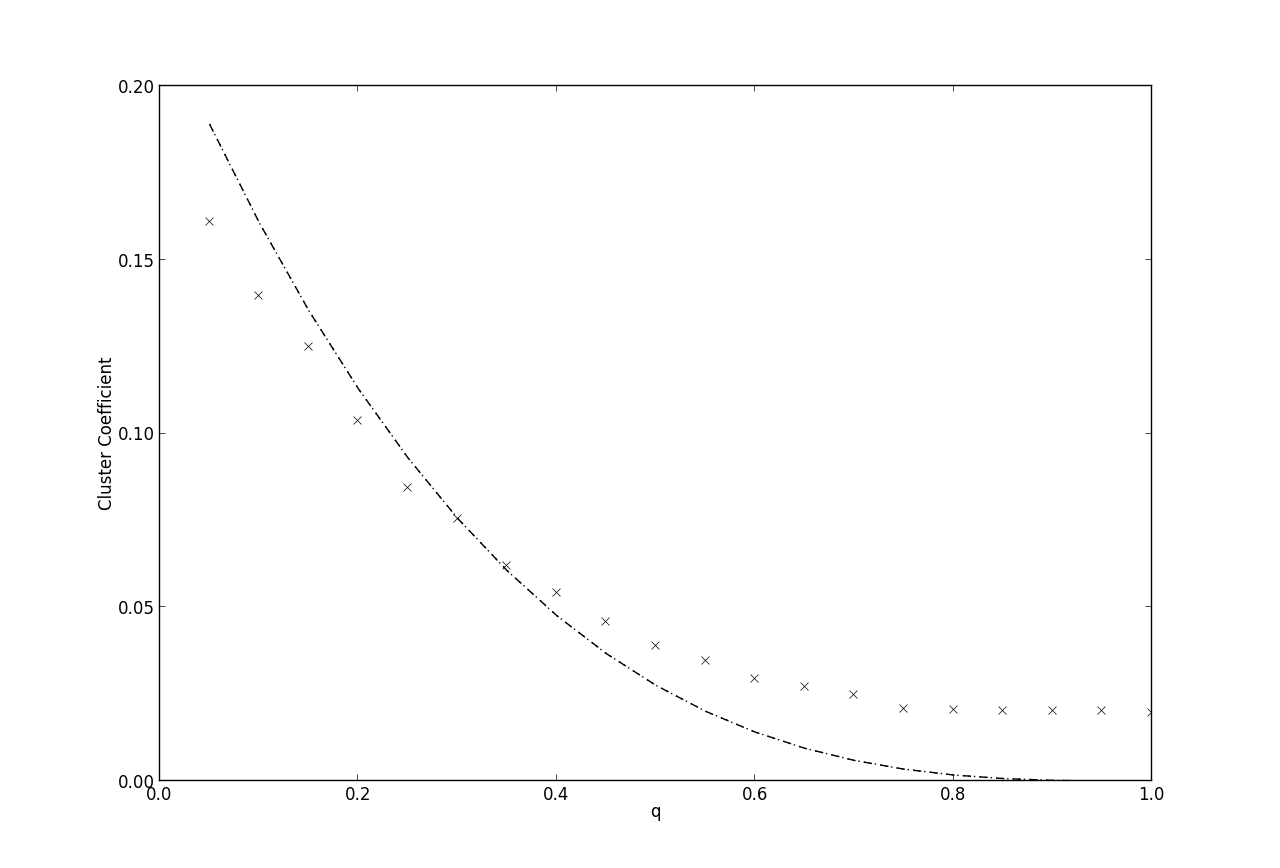}
    \captionof{figure}{Parameters: $n=1000$, $w=10$ and $p\approx 0.01$. The crosses plot the  clustering coefficient from the simulation of the composite network, WS$\cup$ER. The dashed line is the approximation from Eq. \ref{eq:WS2ER_cc} for various values of WS's rewiring probability $q$. }\label{WS2ER_cc}
\end{center}


\section{General Observations}
Our study is not exhaustive because of the large parameters space. The main constraint on the discussion in the present paper consist in the assumption that the two networks being added together have approximately equal number of vertices and edges.  When this assumption is fulfilled we have obtained the following results.

The first observation is that the union with lattice-like WS, yields a composite network with clustering coefficient higher than the other combinations (Figure \ref{CC_all}). This is because there are more triangles at the vertices in a WS network.

For example the clustering coefficient of $v_{i,er}$, $v_{i,ba}$ and $v_{i,ws}$ be proper lowest-term fraction $^a/_b$, $^c/_d$ and $^e/_f$. The numerator is the number of triangles and the denominator is the number of triples. Since the clustering coefficient of $ER < BA < WS$, we have

\begin{equation*}
\frac{a}{b} < \frac{c}{d} < \frac{e}{f}
\end{equation*}

Assuming few common edge, the clustering coefficient of the composite vertex $v_{i,c}$ is the ratio of the total number of triangles to the total number of triples. i.e. for ER$\cup$BA, the clustering coefficient is $(a+c)/(b+d)$. By the mediant inequality,

\begin{equation*}
\frac{a}{b} < \frac{a+c}{b+d} < \frac{c}{d} < \frac{c+e}{d+f} < \frac{e}{f}
\end{equation*}

Hence unions with lattice-like WS (i.e. $(c+e)/(d+f)$) will yield higher clustering coefficient than other combinations. In fact given equal number of edges, the clustering coefficient of $WS \cup WS$ is greater than $WS \cup BA$, which is greater than $WS \cup ER$ (See Fig. \ref{CC_all}).

\begin{center}
  \includegraphics[width=115mm]{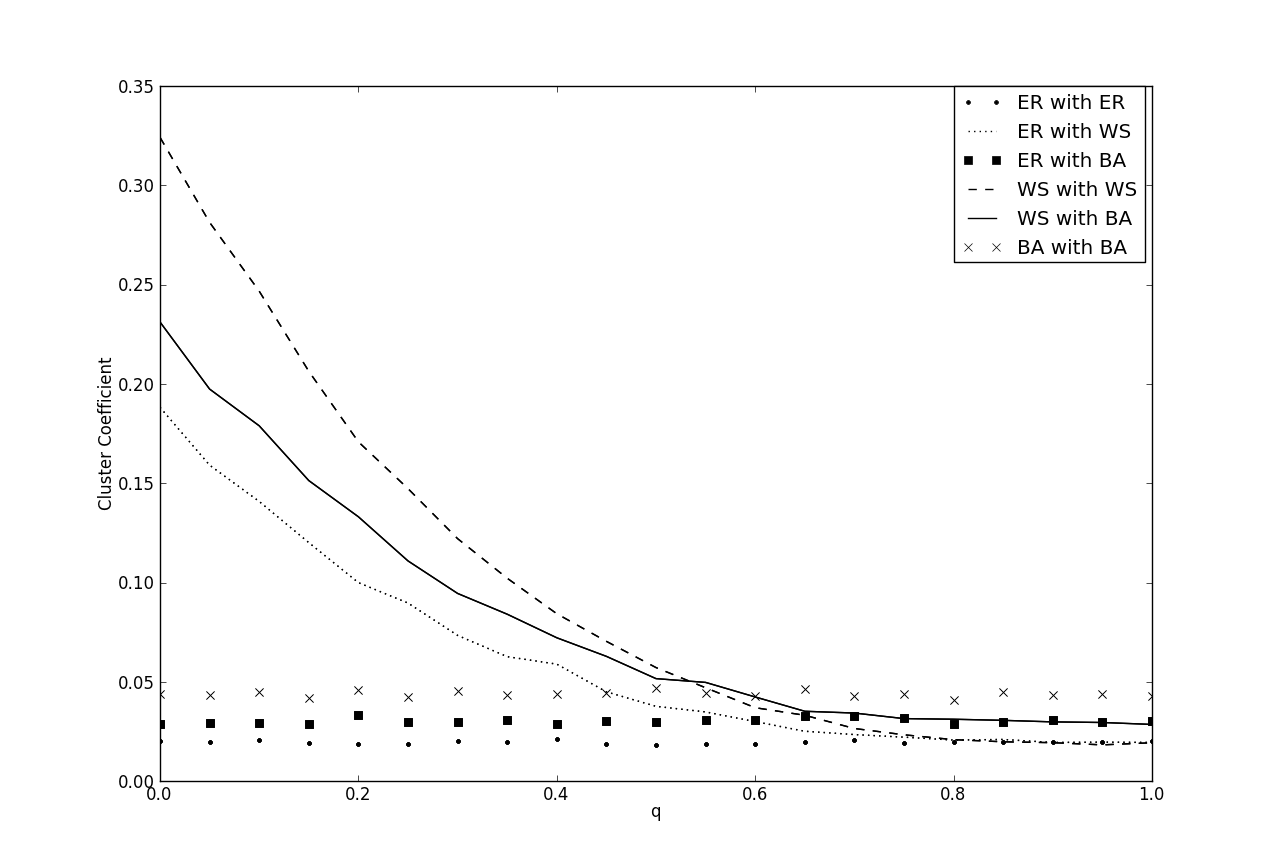}
    \captionof{figure}{Parameters: $n=1000$, $w=10$, $m=5$ and p = $k/(n-1)$. The x-axis varies the rewiring probability ($q$) of WS. For small $q$, combinations with WS are high. Furthermore, the clustering coefficient of $WS \cup WS$ is greater than $WS \cup BA$, which is greater than $WS \cup ER$}\label{CC_all}
\end{center}

From section \ref{sec:DD}, the degree distribution of unions with BA have power-law-like tail distribution. The intuition is that with high probability, the high degree vertices in BA have many common edges with the second network. Hence the change at the high degree vertices is small, allowing the tail distribution of the BA network to dominate. Fig. \ref{DD_all} compares the degree-rank plot of the different union combination.

\begin{center}
  \includegraphics[width=115mm]{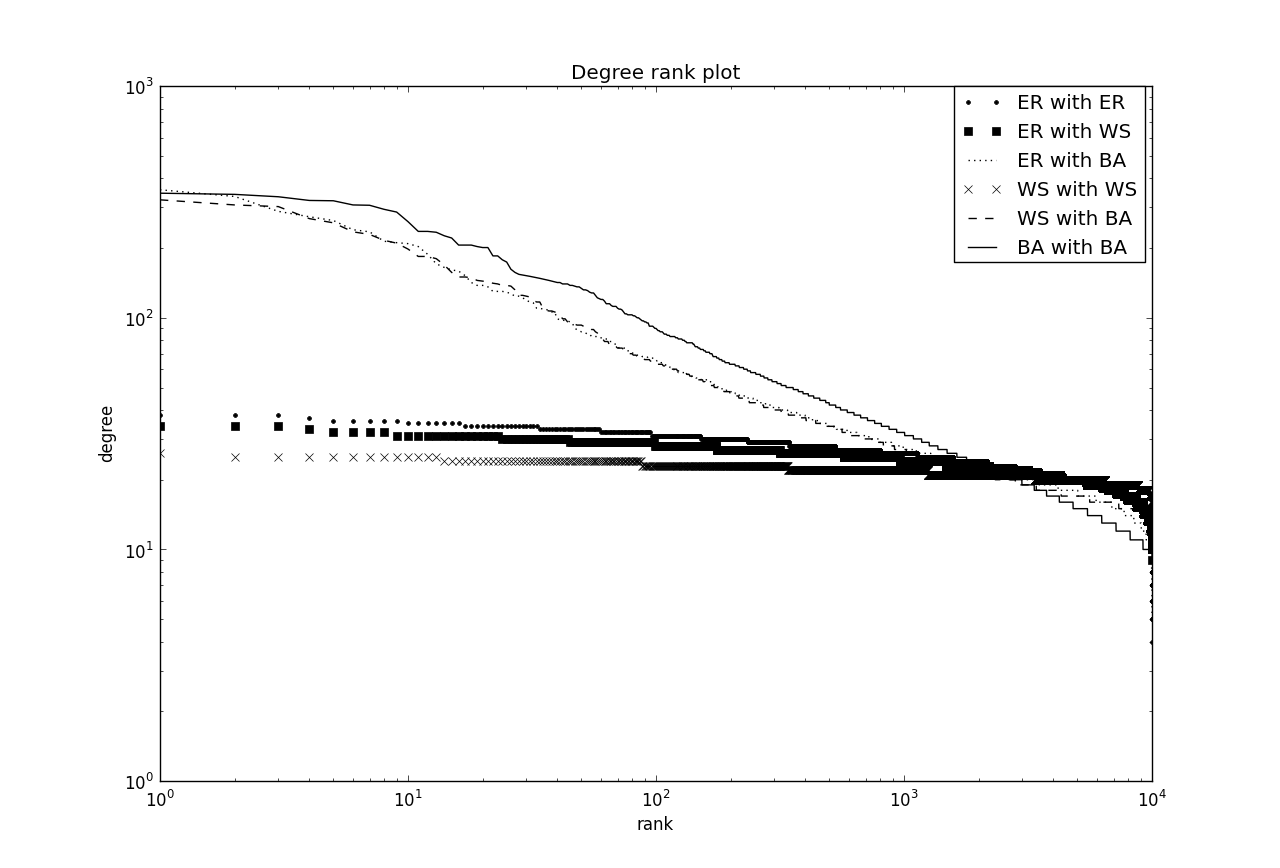}
    \captionof{figure}{Same parameters as Fig. \ref{CC_all} with $q=0.1$. Combinations with BA have power-law-like distribution. Note that $WS \cup BA$ is almost on the same line as $ER \cup BA$. Lastly, $BA \cup BA$ has a slightly different distribution than the other BA combinations.}\label{DD_all}
\end{center}

\section{Discussions and Future Work}
It is common for real-world networks to have similar ``layered dimensions". In temporal networks where connections vary across time, each layer is a time snapshot of the connections of the same vertex set \cite{Holme201297}. The edge union of the network snapshots is analogous to extending the time interval in the sampling, i.e. aggregation of a temporal stream of edges over a fixed number of time distance \cite{conf/icdm/CaceresBG11}. 

Aggregation helps to identify the optimal size of the time interval at which to resolve a temporal network. For example assume that  vertex  $v_i$ and $v_j$ are connected over time $T$ with probability $1-exp(-pT)$ where $p \in [0,1]$ is an arbitrary constant. As we extend (aggregate) the time to $2T$, Eq. (\ref{eq:ER2ER}) gives the edge probability that $v_i$ and $v_j$ are connected over time $2T$ is $1-exp(-2pT)$. This is the same result as we would get from Eq. \ref{eq:ER2ER} if we add the network at $T$ to the one developed over time $T$ to $2T$.

There are cases where the resultant network cannot be easily expressed like the above example or ER $\cup$ ER. This paper address the problem to understand the statistical properties for such cases. In addition, the methods provides a framework for future studies on combinations of other families of network generators like Kronecker Graph \cite{Leskovec}, Exponential Random Graph \cite{frank} and Multifractal Network Generator \cite{Palla}.

Moreover the present investigation can be further developed to study a closely related model,   which considers the percolation across one network layer of an effect induced by an adjacent layer. Such networks are  known as   
 interdependent networks \cite{Buldyrev1,Buldyrev2}.  
  
For example, two networks $G_1(V_1,E_1)$ and $G_2(V_2,E_2)$ form an interdependent network if there exists an edge between a vertex in $V_1$ and $V_2$. Let the set of such edges between $G_1$ and $G_2$ be $E_3$. Hence an interdependent network is formally defined as $G(V_1 \cup V_2, E_1 \cup E_2 \cup E_3)$ if $E_3 \neq \emptyset$. However the resulting network can also be considered as the edge union of 3 networks on the same vertex set.

 From the previous example, we can add isolated vertices to both networks such that they form the same vertex set $V = V_1 \cup V_2$. In this way we extend all three networks in order to consider them as define on a common vertex set. In this way we can obtain that the edge union of all 3 networks form the interdependent network $G(V, E_1 \cup E_2 \cup E_3)$.  

These studies are beyond the scope of this paper but we hope that our paper may help to motivate further theoretical studies on more complex, yet closer to real-world network models.

\section{Acknowledgment}
We would like to thank Tim Evans and the anonymous referees for their comments and suggestions to improve this paper.

\appendix
\section{Simulation Details}
Simulations are implemented in Python 2.7. This is to make use of the existing numerical and network libraries like numpy and NetworkX to minimize implementation errors. The trade off is slower runtime, resulting in smaller sample size.

For simulations, two networks with the same number of vertices are generated with NetworkX network library. The size of the vertex set is chosen such that we can maximize the sample size while computationally feasible (less than 1 hour). For instance the size of the vertex set in Fig. \ref{ER2BA_degree} is 10000 as the iterative method (equation \ref{iterative_ER2BA_dd}) is computationally expensive. In contrast, to generate figure \ref{WS2BA_dd}, it is possible to use vertex set of size 100000.

The size of the edge set also determines the computation costs. In this paper, we chose the minimum size edge set such that the generated networks are connected. The decision for studying connected network is arbitrary from an analysis point of view. However it is useful if we want to study other metrics like average shortest distance in the future.

To ensure ER is a connected network, we need the probability, $p$, that two vertices are connected to satisfy $p > \ln{n}/n$. Furthermore in this paper we only consider the case where both networks have approximately equal number of edges. Thus for WS and BA to have approximately equal edge set size as ER, $w \approx (n-1)p$ and $m \approx (n-1)p/2$ respectively.

We note that unequal size edge sets introduce additional complexity to the study, since the larger edge set will dominate the resultant network. Ignoring this aspect significantly reduces the complexity of the computation.

\end{document}